\documentclass{article}
\usepackage{mwe}

\usepackage{graphicx}
\usepackage{subcaption}
\usepackage{amsmath}
\usepackage{amsfonts}
\usepackage{geometry}
\usepackage{amsthm}
\usepackage{graphicx}

\makeatletter

\numberwithin{figure}{section}
\theoremstyle{plain}
\newtheorem{thm}{\protect\theoremname}
  \theoremstyle{plain}
  \newtheorem{lem}{\protect\lemmaname}
  \theoremstyle{plain}
  \newtheorem{defn}{\protect\definitionname}
  \theoremstyle{plain}
  
  \theoremstyle{plain}
  \newtheorem{cor}{\protect\corollaryname}

  \theoremstyle{plain}

\@ifundefined{showcaptionsetup}{}{%
 \PassOptionsToPackage{caption=false}{subfig}}

\makeatother

  \providecommand{\propname}{Proposition}
  \providecommand{\corollaryname}{Corollary}
  \providecommand{\definitionname}{Definition}
  \providecommand{\factname}{Fact}
  \providecommand{\lemmaname}{Lemma}
\providecommand{\theoremname}{Theorem}
\providecommand{\remarkname}{Remark}

\begin{document}

\title{The Quantum Theil Index: Characterizing Graph Centralization using von Neumann Entropy}

\author{{
\sc David E. Simmons}$^*$,\\[2pt] 
{\sc Justin P. Coon}\\[2pt]
Department of Engineering Science \\University of Oxford \\Oxford, UK, OX1 3PJ\\ 
$^*${ {Corresponding author: david.simmons@eng.ox.ac.uk}} \\
{ {justin.coon@eng.ox.ac.uk}}\\[6pt]
{\sc and}\\[6pt]
{\sc Animesh Datta} \\[2pt]
Department of Physics \\University of Warwick \\Coventry, UK, CV4 7AL\\
 {animesh.datta@warwick.ac.uk}}

\maketitle
\begin{abstract}
{We show that the  von Neumann entropy (from herein referred to as the von Neumann index) of a graph's trace normalized combinatorial Laplacian provides structural information about the level of centralization across a graph. This is done by considering the Theil index, which is an established statistical measure used to determine levels of inequality across a system of `agents', e.g., income levels across a population.
Here, we establish a Theil index for graphs, which provides us with a macroscopic measure of graph centralization.  
Concretely, we show that the von Neumann index can be used to bound the graph's Theil index, and thus we provide a direct characterization of graph centralization via the von Neumann index. Because of the algebraic similarities between the bound and the Theil index, we call the bound the von Neumann Theil index. 
We elucidate our ideas by providing examples and a discussion of different $n=7$ vertex graphs. We also discuss how the von Neumann Theil index provides a more comprehensive measure of centralization when compared to traditional centralization measures, and when compared to the graph's classical Theil index. This is because it more accurately accounts for macro-structural changes that occur from micro-structural changes in the graph (e.g., the removal of a vertex). Finally, we provide future direction, showing that the von Neumann Theil index can be generalized by considering the R\'enyi entropy. We then show that this generalization can be used to bound the negative logarithm of the graph's Jain fairness index.}
{Graph; structure; von Neumann; entropy; complexity; centralization.}
\\
2000 Math Subject Classification: 34K30, 35K57, 35Q80,  92D25
\end{abstract}

\section{Introduction}

Recently, there have been many attempts to increase our understanding of the structure of graphical models by studying entropic quantitites relating to the graph's underlying constructs. Many approaches to studying graph entropy have been proposed, and to the best of the authors' knowledge, it is not clear if/how many of these entropic quantities relate to each other. For example, in \cite{dehmer2008information} an information function was defined on each of the vertices of a graph in order to associate an entropic quantity with the network; in \cite{anand2011shannon}, the entropic properties of graph ensembles have been studied; while in \cite{coon2016topological,artacoon2016topological} Shannon entropy was used to study the topological uncertainty for networks embedded within a spatial domain. There have also been studies of the von Neumann entropy\footnote{If we define the von Neumann entropy $H\left(G\right)$ of a graph to be the von Neumann entropy of its trace normalized Laplacian, then - for unions of disjoint graphs we will not have the sub-additivity property that is typically required from entropic quantity be satisfied: i.e., the von Neumann entropy of the union may be larger than the the sum of the von Neumann entropies of its disjoint parts. This point was discussed in \cite{de2015structural} and can also be seen by observing Proposition 1.5 of \cite{dairyko2017note} and considering the disjoint union of two identical graphs, each with von Neumann entropy greater than one. Consequently, from herein we refer to the   graph's von Neumann entropy as its von Neumann index.} \cite{von1955mathematical} (from herein referred to as the von Neumann index) of graphs \cite{anand2011shannon,braunstein2006laplacian,passerini2008neumann,anand2009entropy,
de2016interpreting,dairyko2017note,belhaj2016weighted,dutta2016graph,dutta2017zero,simmons2017symmetric,giovannetti2011kirchhoff}. This work focuses on such a study. 

\subsection{Graphs and the von Neumann Index}

The mathematical theory of quantum mechanics allows us to view quantum states of finite-dimensional systems as Hermitian, positive semi-definite matrices with unit trace \cite{von1955mathematical}. It is also known that the combinatorial Laplacian of a graph is a symmetric positive semi-definite matrix. Thus, normalizing this matrix by its trace allows us to associate a quantum state with a graph, \cite{braunstein2006laplacian}. 
A natural next step is to study how the structure of the graph is related to the von Neumann index of the graph's trace normalized Laplacian, \cite{braunstein2006laplacian}. This idea was introduced in \cite{braunstein2006laplacian}. To the best of the authors' knowledge, a precise relation is yet to be established, \cite{du2010note}, and the goal of this work is to shed further light on this topic. 

Let us discuss related work in this area. In \cite{passerini2008neumann,du2010note} it was noted that the von Neumann index may be interpreted as a measure of network regularity. In \cite{anand2009entropy}, it was then shown that for scale free networks the von Neumann index of a graph is linearly related to the Shannon entropy of the graph's ensemble. Correlations between these entropies were observed when the graph's degree distribution displayed heterogeneity in 
\cite{anand2011shannon}. 
 In \cite{dairyko2017note}, it was shown that the von Neumann index of an $n$ vertex graph is almost always minimized by the star graph. In that work, the star graph was also conjectured to minimize the entropy over connected graphs. Other studies of the von Neuamnn entropy of Laplacians have been performed in \cite{de2016interpreting,belhaj2016weighted,dutta2016graph,
 dutta2017zero,simmons2017symmetric,de2016spectral}. Also of interest is \cite{giovannetti2011kirchhoff}, in which it is shown that the quantum relative entropy of a graph's Laplacian with another particular quantum state can be related to the number of spanning trees of the graph.

\subsection{Graphs, Centrality, and Centralization\label{sec:centralization}}

Centrality in the context of graphical models was initially investigated in the pioneering work of \cite{bavelas1950communication}. It is concerned with understanding and measuring the \emph{importance} or \emph{centrality} of \emph{particular} vertices within a graph. It is therefore related to the microscopic properties of the graph. The method that is used to define the importance of a vertex is, of course, dependent upon the nature of the underlying system that the graph represents. Because of this, there are many different measures of vertex centrality that have been proposed in the literature \cite{newman2010networks}, e.g., authority centrality, betweenness centrality, closeness centrality, degree centrality, eigenvector centrality, flow betweenness, hub centrality, and Katz centrality. In \cite{qi2012laplacian}, Laplacian centrality was presented. Laplacian centrality is a vertex specific measure that is obtained by calculating the normalized Laplacian energy of the graph before and after a vertex is removed, and computing the difference between these two values.

Related to centrality is the concept of graph centralization. Graph centralization is concerned with understanding the overall structure of the graph, and determining whether it possesses a high or low level of centralization \cite{freeman1978centrality}.  It is therefore related to the macroscopic properties of the graph. As with centrality, there are different ways that graph centralization can be measured. Freeman \cite{freeman1978centrality} showed that any measure of vertex centrality can be extended to a measure of graph centralization. He constructed three centralization measures, each taking values between $0$ and $1$.
Using his measures, he demonstrated that the star graph obtains the maximum possible value, while the complete graph achieves the minimum possible value. A more in depth discussion of graph centralization can be found in \cite{scott2017social}.

In this work, we focus on a novel measure of graph centralization. This measure is not explicitly derived from a vertex centrality measure as per Freeman's approach, \cite{freeman1978centrality}. Instead, it is related to the Theil index \cite{theil1972statistical} (introduced below) and the graph's von Neumann index. To the best of the authors' knowledge, this is the first time such a measure has been proposed.

\subsection{The Theil Index\label{sec:Theilindex}}

The Theil index $T$ has been a topic of interest since its inception in $1972$ \cite{theil1972statistical}. The Theil index  \cite{theil1972statistical,conceiccao2000young} is an established observable that has been used to measure, e.g., inequality within economies, as well as racial segregation within social environments. Formally, the Theil index is defined on a population of $n$ \emph{agents}, each with some \emph{characteristic} $x_i\geq 0$, $i\in\{1,\dots,n\}$.  The Theil index is then given by \cite{conceiccao2000young}
\begin{equation}
T := \frac{1}{n} \sum_{i=1}^{n} \frac{x_i}{\mu}\log\left(  \frac{x_i}{\mu} \right),\label{eq:Theilindex1}
\end{equation}
where $\mu$ is defined to be 
\begin{equation}
\mu := \frac{1}{n} \sum_{i=1}^{n}x_i.\label{eq:Theilindex1mu}
\end{equation} 
The Theil index can be used to characterize the uniformity of the characteristics among the agents, and is bounded according to $$0\leq T  \leq \log n.$$If $ T =0$ then the characteristics are uniform among the agents; if  $ T  =\log n$ then $x_i \neq 0$ for only one $i$.
Some examples of characteristics from the literature are income \cite{novotny2007measurement}  or total wetland \cite{sampath1988equity}.

Through inspection, it is easy to see that the Theil index is closely related to the Shannon entropy. Specifically, $x_i/(\mu n)$ defines a probability distribution on the agents (because $\sum_ix_i/(\mu n) =1$ and $x_i/(\mu n)\geq 0$), so that we can define the Shannon entropy \cite{cover2012elements}
\begin{equation}
H(X) := - \sum_{i=1}^{n} \frac{x_i}{\mu n}\log\left(  \frac{x_i}{\mu n} \right),
\end{equation}
where $  X$ is the random variable with distribution defined by $x_i/(\mu n)$.
The Theil index can then be written as 
\begin{equation}
T = \log n - H.
\end{equation}
Alternatively, we may write
\begin{equation}
T = S\left( X\Vert  {U}_n \right)  ,\label{eq:relativeTheil}
\end{equation}
where $S\left(  A\Vert  {B}  \right) $  denotes the relative entropy \cite{wilde2011classical} between two random variables $  A$ and $  B$, and $  X$ is the random variable described immediately above while $U_n$ is the discrete uniform random variable on $n$ points. Equation \eqref{eq:relativeTheil} provides an information theoretic interpretation of the Theil index as the average number of extra bits of information required to encode a source distributed according to $X$ given that we assume it to be uniformly distributed (the logarithm being base two here). This also adds to our intuition of the Theil index as a measure of uniformity/inequality.
 
\subsection{Contributions}

In this work, we provide an interpretation of the von Neumann index of the graph's Laplacian as a measure of the graph's centralization\footnote{Our work differs from \cite{qi2012laplacian}, in which centrality was established from properties of the graph's Laplacian. First, \cite{qi2012laplacian} proposed a centrality measure, not a graph centralization measure (for a discussion of centrality and centralization, see section \ref{sec:centralization}). Second, that work did not consider the von Neumann entropy when establishing their centrality measure.}. 
This is done by defining a Theil index for the graph in terms of the degrees of the graph's vertices (equation \eqref{eq:TheildegreeEnt} below), which we can use to determine the level of centralization across the vertices of the graph. Our main result is to show that the graph's von Neumann index can be used to bound the graph's Theil index  (Theorem \ref{lem:degreevon} below). The bound is a particular instance of the quantum relative entropy. Because of the relationship between relative entropy and the Theil index \eqref{eq:relativeTheil}, this motivates us to refer to the bound as the graph's von Neumann Theil index (Definition \ref{def:quantumTheil} below). A corollary to our main result (Corollary \ref{cor4}) also determines a relationship between graph von Neumann entropy and the graph entropy considered in \cite{dehmer2008information,cao2014extremality}. To elucidate our ideas, we provide examples and a discussion of different $n=7$ vertex graphs in section \ref{sec:examples}. In these examples, we also consider Freeman's centralization measures \cite{freeman1978centrality}.
 

%

\section{Preliminary Material\label{sec:intro}}

The role of this section is to present background material on spectral graph theory and density operators, and to present the concept of   Theil indices for graphs in terms of the graph's entropy  \cite{dehmer2008information}.

\subsection{Spectral graph theory and density operators}

Consider an undirected graph $G=\left(V,E\right)$, which consists of a vertex set 
\begin{equation}
V=\left\{ v_{1},\dots,v_{n}\right\} , 
\end{equation}
with $\left|V\right|=n$, and an edge set 
\begin{equation}
E=\left\{ e_{1},\dots,e_{m}\right\} ,
\end{equation}
 with $\left|E\right|=m.$ Here, we assume that the element of the edge set (say the $k$th element) connecting vertex pair $\left(i,j\right)$ can be written as \begin{equation}e_{k}=v_{i,j}.\label{eq:edge_vertexeq}\end{equation} The degree matrix of the graph is given by 
\begin{equation}
\Delta\left(G\right)=\mathrm{diag}\left\{ d_{1},\dots,d_{n}\right\} ,
\end{equation}
where $d_{i}$ is the degree of $v_{i}$. The $\left(i,j\right)$th element of a graph's adjacency matrix is given by
\begin{equation}
\left[A\left(G\right)\right]_{i,j}=1
\end{equation}
if vertex pair $\left(i,j\right)$ is connected by an edge. The Laplacian of the graph is then defined to be  \cite{newman2000laplacian}
\begin{equation}
L\left(G\right)=\Delta\left(G\right)-A\left(G\right).
\end{equation}

A graph's Laplacian is useful when studying the connectedness properties of the graph. For example, the number of zero eigenvalues of the Laplacian's spectrum is known to categorize the number of disconnected regions within a graph \cite{newman2000laplacian}. Also, Kirchoff's theorem tells us that the (absolute value of the) determinant of any first order minor associated with a graph's Laplacian is the number of spanning trees of the graph \cite{godsil2013algebraic}.  

More, recently, graph Laplacians have been studied using tools from quantum information theory. Specifically, we can view quantum states as positive semi-definite matrices with unit trace. It is well known that the Laplacian of a graph is a Hermitian positive semi-definite matrix. Thus, after normalizing this matrix by its trace we obtain a new matrix that has all the properties of a density matrix (e.g., a quantum state):
\begin{equation}
\rho_G = \frac{L\left(G\right) }{\mathrm{Tr}\{L\left(G\right)\}}.\label{eq:vertexstate}
\end{equation}
The idea of measuring the von Neumann index of \eqref{eq:vertexstate} was originally posed in \cite{braunstein2006laplacian}:
\begin{equation}
H(G) = H\left(\rho_G \right) :=-\mathrm{Tr}\left\{\rho_G\log\rho_G\right\}.\label{eq:graphvonneumannentropy}
\end{equation}
Note, throughout the paper, we use $H(G)$ and $H(\rho_G)$ interchangeably. With \eqref{eq:graphvonneumannentropy}, we can identify the quantum relative entropy between $\rho_G$ and $I/n$, which is given by
\begin{equation}
D\left( \rho_G \Vert I_n  \right) := \mathrm{Tr}\left\{  \rho_G\left( \log \rho_G - \log I/n  \right) \right\} = \log n - H(G).\label{eq:graphrelativenentropy}
\end{equation}
The main result of this work will be to show that $D\left( \rho_G \Vert I_n  \right)$ bounds a centralization measure (the graph's Theil index, discussed below) on the network. Because of the relationship between relative entropy and the Theil index \eqref{eq:relativeTheil}, our main result motivates us to refer to the bound as the graph's von Neumann Theil index (Definition \ref{def:quantumTheil} below). 

 \subsection{The Theil Index of a Graph\label{sec:graphTheil}}

As was discussed in the introduction, the Theil index $T $ assigns a value to a set of characteristics $\{x_1,\dots,x_n\}$ (see \eqref{eq:Theilindex1} and \eqref{eq:Theilindex1mu}).
In light of this, here we assign a Theil index to a graph by extracting characteristics from the graph's vertices. Specifically, we define a functional 
\begin{equation}
f\; : \; V \mapsto \mathbb{R}^+,\label{eq:functional}
\end{equation} 
on the vertices of an $n$ vertex graph, $G$. If we then let the $n$ characteristics be
$$x_{G,i} := f(v_i),\;i\in\{1,\dots,n\},$$
we obtain the Theil index of the graph $G$
\begin{equation}
T_f(G) := \frac{1}{n} \sum_{i=1}^{n} \frac{x_{G,i}}{\mu}\log\left(  \frac{x_{G,i}}{\mu} \right),
\end{equation}
where 
\begin{equation}
\mu := \frac{1}{n} \sum_{i=1}^{n}x_{G,i}.
\end{equation} 
 
It is important to acknowledge that our work does not establish a relationship between this Theil index and the graph's von Neumann index when $x_{G,i}$ is arbitrary. Instead, we establsih such a relationship when $x_{G,i}$ takes a specific form determined by the degrees of each of the graph's vertices. It is also important to acknowledge that this is \emph{not} the first time such a functional $f$ has been proposed in the literature. In particular, it was used in \cite{dehmer2008information} to define a concept of graph entropy. Specifically, \cite{dehmer2008information} gave the following definition.
\begin{defn}\label{def:genent}
Let $G$ be a graph on $n$ vertices. Let $f$ be the functional \eqref{eq:functional}, which assigns an arbitrary positive number to each vertex $v_i$. Then the entropy of the graph with respect to $f$ is defined to be
\begin{equation}
H_f\left( G \right) := - \sum_{i=1}^n \frac{f(v_i)}{\sum_j f(v_j)} \log \frac{f(v_i)}{\sum_j f(v_j)}.
\end{equation}
\end{defn}

With this definition, we can now write the Theil index of the graph as 
\begin{equation}
T_f(G) = \log n - H_f\left( G \right).\label{eq:generalTheilindex}
\end{equation}
One natural example of a functional is the vertex degree power,  
\begin{equation}f(v_i) = d_i^k,\end{equation}
 $k\in\mathbb{R}$,  which was originally considered in \cite{cao2014extremality}. Here, we call the corresponding entropy the \emph{degree entropy}, which is given by
\begin{equation}
H_{d,k}\left( G \right) := -\sum_i \tilde{d}_i^k \log \tilde{d}_i^k, \label{eq:degreeentropy}
\end{equation}
where 
\begin{equation}
\tilde{d}_i^k := \frac{d_i^k}{\sum_j d_j^k}\label{eq:relativedegreedist}
\end{equation}
is the $k$th power relative degree of vertex $i$. We can then write the Theil index corresponding to the graph's degree entropy as 
\begin{equation}
T_{d,k} \left( G \right) := \log n - H_{d,k}(G).\label{eq:TheildegreeEnt}
\end{equation}
 
\subsection{The Theil index and Graph Centralization\label{sec:TheilandCent}} 
 
Naturally, we can use the Theil index of the graph corresponding to the degree entropy, \eqref{eq:TheildegreeEnt},  to measure the level of centralization present across a network. Specifically, because the Theil index is a measure of relative entropy, we can use it to measure how `close' the relative degree distribution \eqref{eq:relativedegreedist} is to the uniform distribution:  regular graphs have a uniform degree distribution, while star graphs may be perceived as having the `least uniform' degree distribution. It is easy to show that these examples, respectively, induce minimum and maximum values for the Theil index. Intuitively, the complete graph (which is a $k$-regular graph) may be thought of as being the least centralized graph, while the star graph may be thought of as being the most centralized graph.

\section{The Graph Theil Index and von Neumann Entropy}

Having developed the Theil index of the graph \eqref{eq:generalTheilindex}, we now present the following Theorem, which provides a relationship between the graph's Theil index\footnote{Specifically, the Theil index corresponding to its degree entropy \eqref{eq:TheildegreeEnt}.} and the quantum relative entropy defined in \eqref{eq:graphrelativenentropy}. From section \ref{sec:TheilandCent}, this then identifies a relationship between the quantum relative entropy of the graph and its centralization. Because of these intricate relationships, and the relationship that the Theil index has with relative entropy (see \eqref{eq:relativeTheil}), before presenting the theorem we first define the von Neumann Theil index of the graph.
\begin{defn}\label{def:quantumTheil}
Let $\rho_G$ be given by \eqref{eq:vertexstate} and $I$ be an equivalently dimensioned identity matrix. The von Neumann Theil index of the graph $T_{Q}(G)$ is defined to be the quantum relative entropy between $\rho_G$ and $I/n$, i.e.,
\begin{equation}
T_{Q} \left( G \right) := D\left( \rho_G \Vert I/n\right) := \mathrm{Tr}\left\{\rho_G \left(   \log \rho_G - \log I/n\right) \right\} = \log n - H (G).\label{eq:Theilquant}
\end{equation}
\end{defn}
\begin{thm}\label{lem:degreevon}
Let $G$ be a connected graph on $n$ vertices, $T_{d,k}\left( G \right)$ be the graph's Theil index corresponding to its degree entropy \eqref{eq:TheildegreeEnt}, and $T_Q(G)$ be the graph's von Neumann Theil index \eqref{eq:Theilquant}. Also, let $M :=  \{ d_i : d_i = \max_j \{ d_j\}\}$, so that $\vert M\vert$ is the multiplicity of $\max_i \{d_i\}$.
Exactly one of the following statements holds:
\begin{description} 
\item[\emph{A:}] 
$\log n - \log \vert M\vert \geq T_{Q}$ if and only if there exists $k$ such that
\begin{equation}
T_{d,1} \left( G \right) \leq T_Q(G) \leq T_{d,k'} \left( G \right)\;\forall \;k'\geq k.\label{eq:lowerboundvonneumannentropy}
\end{equation}
\item[$\quad$] Furthermore,
\begin{equation}
\mathrm{if}\quad\frac{\left( \sum_i d_i\right)^2}{\sum_i d_i + \sum_i d_i^2} \geq \vert M\vert \quad\mathrm{then\quad} \log n - \log \vert M\vert \geq  T_{Q}.\label{eq:renyi2entropyMrelation}
\end{equation}
\item[\emph{B:}]\label{statementA}  $\log n -\log \vert M\vert \leq  T_{Q}$ if and only if  
\begin{equation}
T_{d,k} \left( G \right) \leq T_Q(G) \;\forall \;k.\label{eq:upperboundvonneumannentropy}
\end{equation}
\end{description}
\end{thm}
\begin{proof}
The proof is given in section \ref{lem:degreevonproof} below.
\end{proof}

\subsection{Discussion of Theorem \ref{lem:degreevon}}

Theorem \ref{lem:degreevon} deserves discussion. First and foremost, this theorem highlights a fundamental relationship between graph von Neumann entropy (and the von Neumann Theil index) and graph centralization. Specifically, it tells us that $T_Q(G)$ (see \eqref{eq:graphrelativenentropy}) \emph{always} upper bounds the centralization measure when $k=1$. Because $T_{d,k} \left( G \right)$ is a non decreasing function of $k$ (Lemma \ref{lem:decreasingfuncofk}), it also tells us that, given $\log n - \log \vert M\vert \geq T_{Q}(G)$, $T_Q(G)$ will (for a specific $k>1$) equal the centralization measure exactly. Equation \eqref{eq:renyi2entropyMrelation} provides information about when we can be sure  $\log n - \log \vert M\vert \geq T_{Q}$. These ideas are formalized in the following corollary.
\begin{cor}
We have \begin{equation}T_{d,1} (G) \leq T_Q(G).\label{eq:Sd1H}\end{equation}
 Furthermore, if $\log n - \log \vert M\vert \geq T_{Q}$ then there exists $k\in\mathbb{R}^+$ such that
\begin{equation}T_{d,k} \left( G \right) = T_Q(G) .\label{eq:Theuilindex}\end{equation}
\end{cor}


Interestingly, $k$-regular graphs always satisfy the first statement of Theorem \ref{lem:degreevon}. This is because for any graph $G$, $H(G)\leq \log (n-1)$ \cite{passerini2008neumann}, while $\log|M|=\log n$ when  $G$ is $k$-regular.
We therefore reach a second corollary of the theorem.
\begin{cor}\label{cor:kreg}
If $G$ is $k$-regular then
\begin{equation} 0 = T_{d,k}  (G)  \leq T_Q(G)\;\forall\;k\geq1.\label{eq:Sd1H11}
\end{equation}
\end{cor}

It is important to note that Theorem \ref{lem:degreevon} could have been stated as per the following corollary. However, doing so would not have elucidated the connection between von Neumann entropy and graph centralization.
\begin{cor}\label{cor4}
Consider an identical hypothesis as that of Theorem \ref{lem:degreevon}. Then exactly one of the following statements holds:
\begin{description} 
\item[\emph{A}]\label{statementA}  $H(G) \leq \log \vert M\vert$ if and only if  
\begin{equation}
H_{d,k} \left( G \right) \geq   H(G) \;\forall \;k.\label{eq:upperboundvonneumannentropycor}
\end{equation}
\item[\emph{B}] 
$H\left(G\right) \geq \log \vert M\vert$ if and only if there exist $k$ such that
\begin{equation}
H_{d,1} \left( G \right) \geq   H(G) \geq H_{d,k'} \left( G \right)\;\forall \;k'\geq k.\label{eq:lowerboundvonneumannentropycor}
\end{equation}
\end{description}
\end{cor}
Thus, Corollary \ref{cor4} provides an alternative relationship between the relative degree entropy and the von Neumann entropy of a graph.

\section{Proof of Theorem \ref{lem:degreevon}\label{lem:degreevonproof}}

To prove \eqref{eq:lowerboundvonneumannentropy} of statement $\mathbf{A}$, we begin by showing that $H_{d,1}\left( G \right) \geq H(G)$ (i.e., $T_{d,1}(G)\leq T_Q(G)$).  If  we select a set of  positive operator valued measures $\{\Lambda_i\}$, $i\in\{1,\dots,n\}$, such that $\left[ \Lambda_i\right]_{k,k} = 1$ and all other entries are zero, with $\rho_{G} $ given by \eqref{eq:vertexstate} we find that
\begin{equation}
H_{d,1} (G)  =  -\sum_i\mathrm{Tr}\{  \Lambda_i  \rho_{G}    \log \left(  \Lambda_i  \rho_{G} \right)  \}.
\end{equation}
The first inequality of \eqref{eq:lowerboundvonneumannentropy} then follows from \cite[Theorem 11.1.1]{wilde2011classical}, which shows that for an arbitrary quantum state $\rho$ the von Neumann entropy can be written as
\begin{equation}
H(\rho) = \min_{\{\Lambda_y\}} \left[  -\sum_y\mathrm{Tr}\{  \Lambda_y \rho    \log \left(  \Lambda_y \rho \right)\} \right],
\end{equation}
where the minimum is taken over all rank $1$ positive operator valued measures.
The second inequality of \eqref{eq:lowerboundvonneumannentropy} follows immediately from Lemma \ref{lem:decreasingfuncofk} below.

To prove \eqref{eq:renyi2entropyMrelation} of statement $\mathbf{A}$, we consider the R\'enyi 2 entropy of $G$, which is given by \cite[Equation (2)]{dairyko2016note}
\begin{equation}
H_2(G) = \log \left( \frac{\left( \sum_i d_i\right)^2}{\sum_i d_i + \sum_i d_i^2} \right).\label{eq:Renyi2graphentropy}
\end{equation}
We also note that $H_2(G)\leq H(G)$ $\forall \;G$, \cite{dairyko2016note}. Consequently, the following implications hold:
\begin{align}
  \frac{\left( \sum_i d_i\right)^2}{\sum_i d_i + \sum_i d_i^2}   \geq   \vert M\vert\Leftrightarrow  \log\frac{\left( \sum_i d_i\right)^2}{\sum_i d_i + \sum_i d_i^2}     \geq   \log\vert M\vert\Leftrightarrow H_2(G)& \geq   \log\vert M\vert \label{eq:renyi2entG}
\Rightarrow H(G)\geq   \log\vert M\vert .
\end{align}
where the second bidirectional implication follows from \eqref{eq:Renyi2graphentropy}.

 Statement $\mathbf{B}$ follows immediately from Lemma \ref{lem:decreasingfuncofk} below.

\begin{lem}\label{lem:decreasingfuncofk}
For all graphs $G$ on $n$ vertices, the Theil degree index of the graph $ T_{d,k}$ is a nondecreasing function of $k$, i.e.,
\begin{equation}
\frac{\mathrm{d} T_{d,k} }{\mathrm{d} k} \geq  0. \label{eq:derivativeHdk}
\end{equation} 
Also, with $M =  \{ d_i : d_i = \max_j \{ d_j\}\}$, we have
\begin{equation}
\lim_{k\to\infty} T_{d,k} \left(  G \right) = \log n - \log\vert M\vert. \label{eq:limitdegreeentropyk124}
\end{equation}
\end{lem}
\begin{proof}
We begin by proving \eqref{eq:derivativeHdk}. The basis of the proof will be in the construction of the bound in \eqref{eq:longbound} and \eqref{eq:longbound2} below. To aid our construction of the bound, we note that in general
$\frac{\mathrm{d} }{\mathrm{d} k}\left( \frac{d_i^k}{\sum_j d_j^k} \right)$
can be either positive or negative, and for any particular $k = k'$, we can construct two sets $\mathcal{N}_{k'}$ and $\mathcal{P}_{k'}$ (corresponding to negative and positive elements, respectively), where
\begin{align}
\mathcal{N}_{k'} =\left\{  i : \left. \frac{\mathrm{d} }{\mathrm{d} k}\left( \frac{d_i^k}{\sum_j d_j^k} \right) \right|_{k=k'}  <0 \right\}\quad \mathrm{and}\quad
\mathcal{P}_{k'} =\left\{ i : \left. \frac{\mathrm{d} }{\mathrm{d} k}\left( \frac{d_i^k}{\sum_j d_j^k} \right)   \right|_{k=k'}   \geq 0 \right\}.\label{eq:posnegsets}
\end{align}
With the above, we are in a position to construct the bound:
\begin{align}
\left. \frac{\mathrm{d} H_{d,k}}{\mathrm{d} k} \right|_{k = k'}=&  \left.-\sum_i \frac{\mathrm{d} }{\mathrm{d} k}\left( \frac{d_i^k}{\sum_j d_j^k} \log \frac{d_i^k}{\sum_j d_j^k}\right)  \right|_{k = k'}\\
=& \left.-\sum_i \left[  \frac{\mathrm{d} }{\mathrm{d} k}\left( \frac{d_i^k}{\sum_j d_j^k} \right)  \log \frac{d_i^k}{\sum_j d_j^k}+ \frac{d_i^k}{\sum_j d_j^k} \frac{\mathrm{d} }{\mathrm{d} k}\left(  \log \frac{d_i^k}{\sum_j d_j^k}\right)\right]  \right|_{k = k'} \nonumber\\
=&  \left.-\sum_i \left[ \frac{\mathrm{d} }{\mathrm{d} k}\left( \frac{d_i^k}{\sum_j d_j^k} \right)  \log \frac{d_i^k}{\sum_j d_j^k} + \frac{\mathrm{d} }{\mathrm{d} k}\left(    \frac{d_i^k}{\sum_j d_j^k}\right)\right] \right|_{k = k'}\nonumber\\
=&  -\sum_i  \left. \frac{\mathrm{d} }{\mathrm{d} k}\left( \frac{d_i^k}{\sum_j d_j^k} \right)\right|_{k=k'}  \log \frac{d_i^{k'}}{\sum_j d_j^{k'}}  - \left. \frac{\mathrm{d} }{\mathrm{d} k}\left(  \frac{ \sum_y  d_y^k}{\sum_j d_j^k}\right) \right|_{k=k'}\nonumber\\
=& \left. -\sum_i \frac{\mathrm{d} }{\mathrm{d} k}\left( \frac{d_i^k}{\sum_j d_j^k} \right) \right|_{k = k'}   \log \frac{d_i^{k'}}{\sum_j d_j^{k'}} \nonumber\\
= &   \left. \sum_{i\in\mathcal{N}_{k'}} 
 \frac{\mathrm{d} }{\mathrm{d} k}\left( \frac{d_i^k}{\sum_j d_j^k} \right) \right|_{k = k'} \log \frac{\sum_j d_j^{k'}}{d_i^{k'}}
+\left. \sum_{t\in\mathcal{P}_{k'}}  \frac{\mathrm{d} }{\mathrm{d} k}\left( \frac{ d_t^k}{\sum_j d_j^k} \right)  \right|_{k = k'} \log \frac{\sum_j d_j^{k'}}{d_t^{k'}} \nonumber\\
\leq  & \left.\min_{i\in \mathcal{N}_{k'}} \left\{\log \frac{\sum_j d_j^{k'}}{ d_i^{k'}}\right\}  \sum_{i\in\mathcal{N}_{k'}} 
\left. \frac{\mathrm{d} }{\mathrm{d} k}\left( \frac{d_i^k}{\sum_j d_j^k} \right)  \right|_{k=k'}
+ \max_{i\in \mathcal{P}_{k'}} \left\{ \log \frac{\sum_j d_j^{k'}}{  d_t^{k'}} \right\} \sum_{t\in\mathcal{P}_{k'}}  \frac{\mathrm{d} }{\mathrm{d} k}\left( \frac{ d_t^k}{\sum_j d_j^k} \right) \right|_{k = k'}\label{eq:longbound},
\end{align}
where the final inequality follows because the first and second summations are (by the definition of the sets in \eqref{eq:posnegsets}) negative and positive, respectively.
From Lemma \ref{lem:decreasingfuncofk22}, we have 
$$\max_{i\in\mathcal{P}_{k'}}\left\{\log \frac{\sum_j d_j^{k'}}{ d_i^{k'}}\right\} \leq \min_{i\in\mathcal{N}_{k'}}\left\{\log \frac{\sum_j d_j^{k'}}{ d_i^{k'}}\right\},$$ so that we can extend the final inequality of \eqref{eq:longbound} to
\begin{align}
\left. \frac{\mathrm{d} H_{d,k}}{\mathrm{d} k} \right|_{k = k'} \leq  & \left.\left. \max_{i\in\mathcal{P}_{k'}}\left\{\log \frac{\sum_j d_j^{k'}}{ d_i^{k'}}\right\}  \sum_{i\in\mathcal{N}_{k'}} 
 \frac{\mathrm{d} }{\mathrm{d} k}\left( \frac{d_i^k}{\sum_j d_j^k} \right)  \right|_{k=k'}
+ \max_{i\in\mathcal{P}_{k'}}\left\{\log \frac{\sum_j d_j^{k'}}{ d_i^{k'}}\right\} \sum_{t\in\mathcal{P}_{k'}}  \frac{\mathrm{d} }{\mathrm{d} k}\left( \frac{ d_t^k}{\sum_j d_j^k} \right) \right|_{k = k'}\nonumber\\
= & \max_{i\in\mathcal{P}_{k'}}\left\{\log \frac{\sum_j d_j^{k'}}{ d_i^{k'}}\right\}  \left(\left.\left. \sum_{i\in\mathcal{N}_{k'}} 
 \frac{\mathrm{d} }{\mathrm{d} k}\left( \frac{d_i^k}{\sum_j d_j^k} \right)  \right|_{k'=k}
+  \sum_{t\in\mathcal{P}_{k'}}  \frac{\mathrm{d} }{\mathrm{d} k}\left( \frac{ d_t^k}{\sum_j d_j^k} \right) \right|_{k = k'}\right) \nonumber\\
= & \max_{i\in\mathcal{P}_{k'}}\left\{\log \frac{\sum_j d_j^{k'}}{ d_i^{k'}}\right\} \left.  
 \frac{\mathrm{d} }{\mathrm{d} k}\left(  \frac{\sum_{i } d_i^k}{\sum_j d_j^k} \right)  
  \right|_{k = k'}  =  0 \label{eq:longbound2}.
\end{align}
Substituting this into \eqref{eq:TheildegreeEnt} gives the result.

We now prove \eqref{eq:limitdegreeentropyk124}. For arbitrary $k$, we have
\begin{align}
H_{d,k}\left( G \right) &= -\sum_i \frac{d_i^k}{\sum_j d_j^k} \log \frac{d_i^k}{\sum_j d_j^k}  \\
& \to  -\sum_M \frac{1}{\vert M \vert} \log \frac{1}{\vert M \vert}\\
& = \log \vert M\vert ,
\end{align}
where the second line is obtained by letting $k\to\infty$.
\end{proof}

The following lemma is used within the proof of Lemma \ref{lem:decreasingfuncofk}.
\begin{lem}\label{lem:decreasingfuncofk22}
Let
\begin{align}
\mathcal{N}_{k'} =\left\{  i : \left. \frac{\mathrm{d} }{\mathrm{d} k}\left( \frac{d_i^k}{\sum_j d_j^k} \right) \right|_{k=k'}  <0 \right\}\quad \mathrm{and}\quad
\mathcal{P}_{k'} =\left\{ i : \left. \frac{\mathrm{d} }{\mathrm{d} k}\left( \frac{d_i^k}{\sum_j d_j^k} \right)   \right|_{k=k'}   \geq 0 \right\}.
\end{align}
Then \begin{equation}\max_{i\in\mathcal{N}_{k'}}\{ d_i \} \leq \min_{i\in\mathcal{P}_{k'}}\{ d_i \}.\label{eq:otherbound2}\end{equation}
\end{lem}
\begin{proof}
We have 
$$\left. \frac{\mathrm{d} }{\mathrm{d} k}\left( \frac{d_i^k}{\sum_j d_j^k} \right) \right|_{k=k'}= d_i^{k'} \left( \frac{ \left(   \sum_j d_j^{k'}\right) \log d_i -  \left.\frac{\mathrm{d} }{\mathrm{d}k} \left(  \sum_j d_j^{k}  \right)\right|_{k=k'}}{\left(   \sum_j d_j^{k'} \right)^2}  \right), $$
which is negative if and only if $$\left(   \sum_j d_j^{k'}\right) \log d_i - \left. \frac{\mathrm{d} }{\mathrm{d}{k}} \left(  \sum_j d_j^{k}  \right) \right|_{k=k'}< 0,
\quad \mathrm{if \;and \;only \; if }\quad
 d_i < \exp\left(\frac{\left. \frac{\mathrm{d} }{\mathrm{d}{k}} \left(  \sum_j d_j^{k}  \right) \right|_{k=k'}}{\sum_j d_j^{k'}}\right).$$
Thus, by the definition of the sets $\mathcal{N}_{k'}$ and $\mathcal{P}_{k'}$
$$  \max_{i\in\mathcal{N}_{k'}}\{ d_i \}  < \exp\left(\frac{\left. \frac{\mathrm{d} }{\mathrm{d}{k}} \left(  \sum_j d_j^{k}  \right) \right|_{k=k'}}{\sum_j d_j^{k'}}\right)\leq \min_{i\in\mathcal{P}_{k'}}\{ d_i \}.$$
This proves the lemma.
\end{proof}

\section{Examples and Discussion\label{sec:examples}}

In \cite{freeman1978centrality}, it was argued that the centralization of a graph should index the tendency of a single vertex within the graph to be more central
than  other vertices in the graph. This type of graph centralization measure is typically based on differences between the centrality of the most central
point and that of all others \cite{freeman1978centrality}. Here, we claim that such an approach gets part way to understanding how centralized a graph is, but (in many scenarios) fails to establish the tendency of the graph's centralization to change \emph{significantly} under \emph{small} structural changes. An example of such a graph might be the circle graph, the vertices of which have equal centrality measures throughout. Thus, with the centralization measure proposed by \cite{freeman1978centrality}, the circle graph's centralization will be $0$. However, the removal of any single vertex will result in a path graph, which  canonically possesses significantly more centralization. Interestingly, the Theil index also suffers from this type of problem; however, the von Neumann Theil index does not.

In what follows, we present examples of $10$ common graphs with $n=7$ vertices in FIGs \ref{fig:CD}, \ref{fig:CB}, \ref{fig:TI}, and \ref{fig:QTI}. These figures serve to illustrate the above discussion. Each of the figures is a particular ordering of the $10$ graphs, where the ordering is determined by: 
\begin{enumerate}
\item two centralization measures proposed by Freeman \cite{freeman1978centrality}, based, respectively, on degree centrality (FIG \ref{fig:stargraphCD} to \ref{fig:completegraphCD})
\begin{equation}
C_D = \frac{\sum_{i=1}^n \max\{ d_j \} - d_i}{n^2 -3n + 2} , \label{eq:Freeman}
\end{equation}
and betweenness centrality\footnote{In \eqref{eq:FreemanCB}, $g_i$ is defined to be the betweenness centrality \cite{newman2010networks} of vertex $i$.} (FIG \ref{fig:stargraphCB} to \ref{fig:completegraphCB})
\begin{equation}
C_B = \frac{\sum_{i=1}^n \max\{ g_j \} - g_i}{n^3 - 4n^2 + 5n - 2} ; \label{eq:FreemanCB}
\end{equation}
\item the Theil index \eqref{eq:TheildegreeEnt} with $k=1$ (FIG \ref{fig:stargraph} to \ref{fig:completegraph}); 
\item the von Neumann Theil index \eqref{eq:Theilquant} (FIG \ref{fig:stargraph1} to \ref{fig:completegraph1}).
\end{enumerate}

These orderings provide us with a visualization of a $7$ vertex graph as it becomes more decentralized when measured using each of these metrics. As can be seen from these figures, the different measures impose different orderings on the graphs. However, all measures determine the star graph to be the most centralized and the complete graph to be the least centralized. Note, it was shown in \cite{dairyko2017note} that almost all graphs achieve a larger von Neumann entropy (and hence smaller von Neumann Theil index) than the star graph. It was also conjectured that the star graph minimizes the von Neumann entropy among all connected graphs.

Let us discuss our thesis in more detail. As was discussed in sections \ref{sec:Theilindex} and \ref{sec:graphTheil}, the Theil index depends only on the relative degrees of the vertices, and provides a measure of their skewness. However, it fails to capture richer structural details about the underlying connectedness of the vertices. Similar ideas hold for Freeman's degree centralization measure, and to a lesser extent his betweenness centralization measure. Crucially, they may all be considered as somewhat crude measures of centralization when viewed in the context of graphical models. Our rational for this stems from our observation that the von Neumann Theil index is able to capture the tendency of a graph's macro-structure to change under small micro-structural changes (e.g., vertex removals). The most obvious example of this is the circle and complete graph  (although our arguments can be extended to the other graph's too), which are both $k$-regular, but the vertices of the circle graph are much more dependent on {some} vertices than others. This is not so for the complete graph. Arguably then, the circle graph is considerably more centralized than the complete graph. Despite this, both of Freeman's measures, and the Theil index, rank these two graphs as being the most decentralized out of all graphs, whereas the von Neumann Theil index rightly identifies the more centralized nature of the circle graph while categorizing the complete graph as the most decentralized. Similar discussions can be had towards the other graphs.

The von Neumann Theil index is able to capture the centralization properties that we desire from the Theil index (i.e., it bounds the skew measured by the Theil index, which we interpret as a crude measure of centralization, Theorem \ref{lem:degreevon}), but it also depends on the underlying connectedness of the graph in a way that the other measures do not. In the quantum sense, the von Neumann Theil index is just the quantum relative entropy between the quantum state corresponding to the graph's Laplacian and a maximally mixed quantum state with density operator $I/n$. Consequently, it tells us how `close' the underlying quantum state corresponding to the graph is to the maximally mixed state, i.e., the complete graph. This discussion goes further to explaining our ideas.

 \begin{figure}
 \textbf{Graph ordering according to the degree centralization measure $C_D$ \eqref{eq:Freeman}}\\
    \centering
    \begin{subfigure}{0.19\textwidth}
        \centering
        \includegraphics[width=1\textwidth]{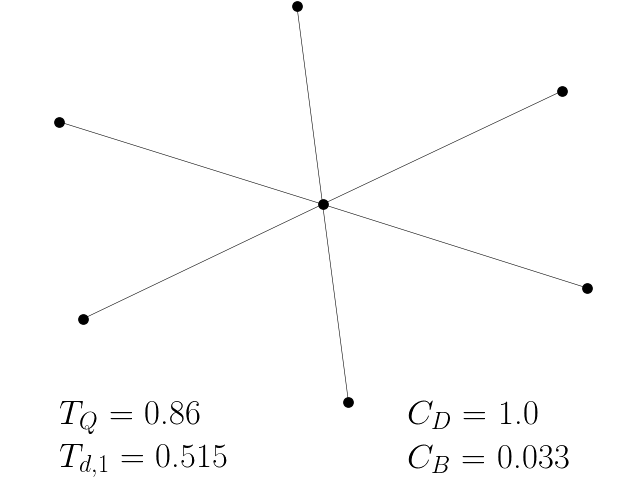} 
        \caption{Star graph. \label{fig:stargraphCD}}
    \end{subfigure}\hfill 
     \begin{subfigure}{0.19\textwidth}
        \centering
        \includegraphics[width=1\textwidth]{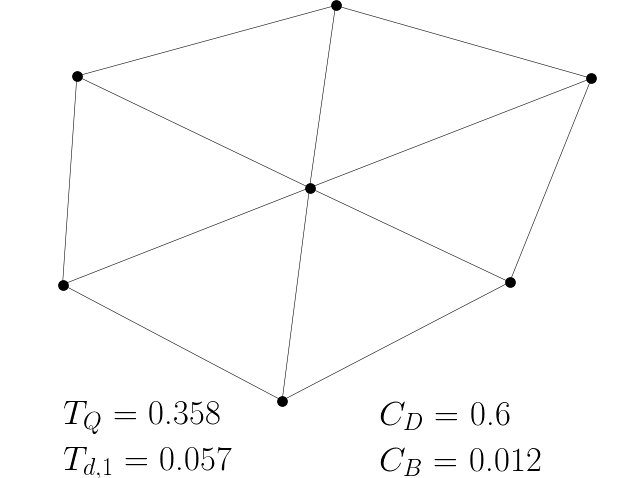} 
        \caption{Wheel graph.}
    \end{subfigure}  
    \begin{subfigure}{0.19\textwidth}
        \centering
        \includegraphics[width=1\textwidth]{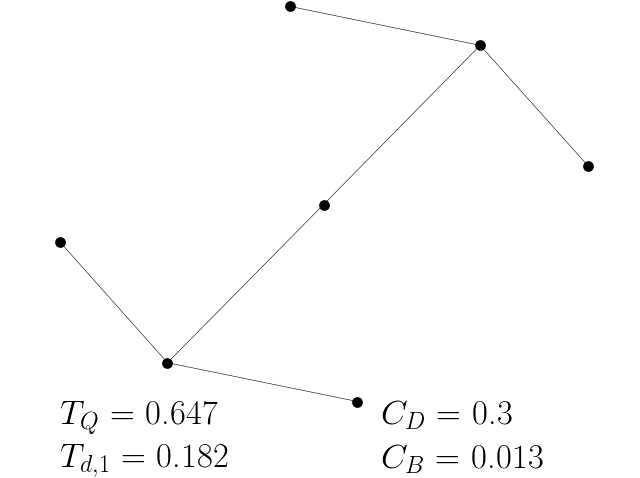} 
        \caption{Balanced tree.}
    \end{subfigure} 
    \hfill  
        \begin{subfigure}{0.19\textwidth}
        \centering
        \includegraphics[width=1\textwidth]{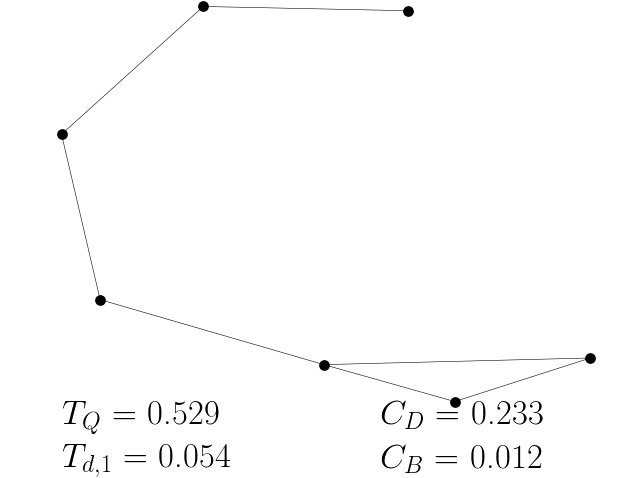} 
        \caption{Lolly pop graph.}
    \end{subfigure}
    \begin{subfigure}{0.19\textwidth}
        \centering
        \includegraphics[width=1\textwidth]{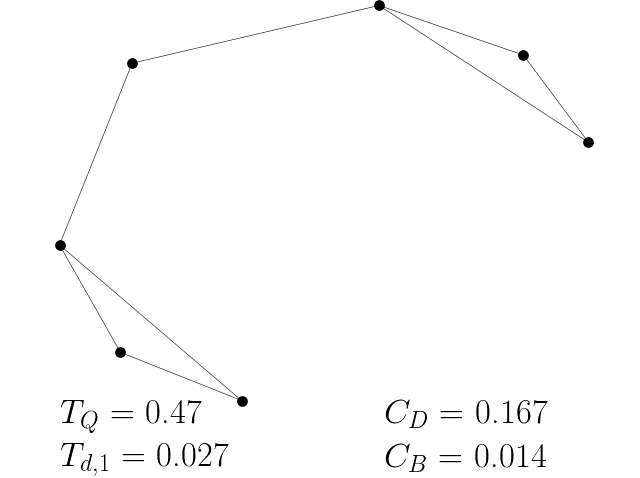} 
        \caption{Barbell graph.}
    \end{subfigure}\hfill 
    \begin{subfigure}{0.19\textwidth}
        \centering
        \includegraphics[width=1\textwidth]{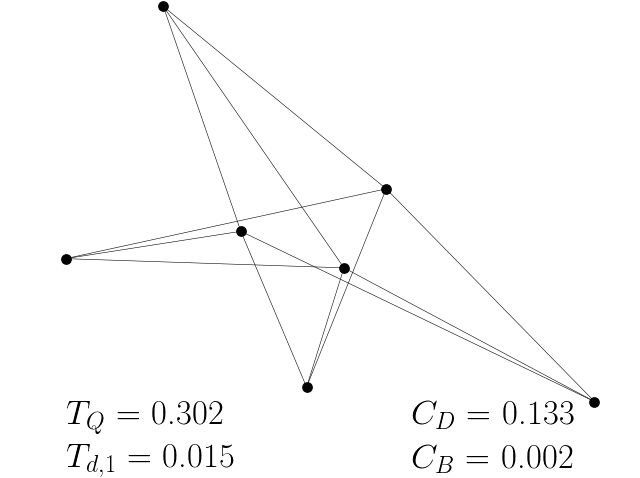} 
        \caption{$(3,4)$-bipartite graph.}
    \end{subfigure}\hfill
    \begin{subfigure}{0.19\textwidth}
        \centering
        \includegraphics[width=1\textwidth]{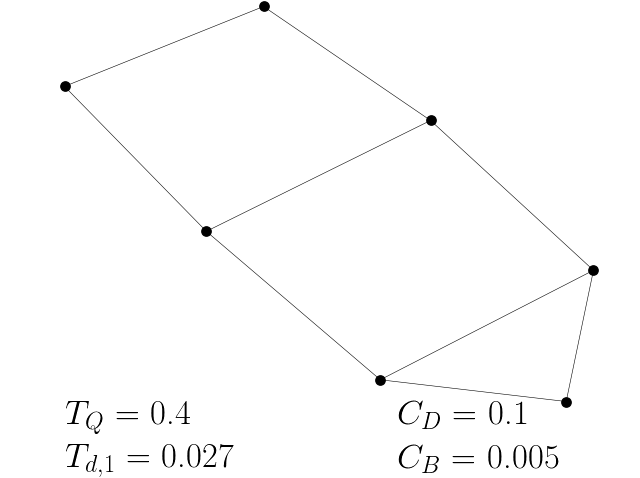} 
        \caption{Two story house graph.}
    \end{subfigure} 
    \begin{subfigure}{0.19\textwidth}
        \centering
        \includegraphics[width=1\textwidth]{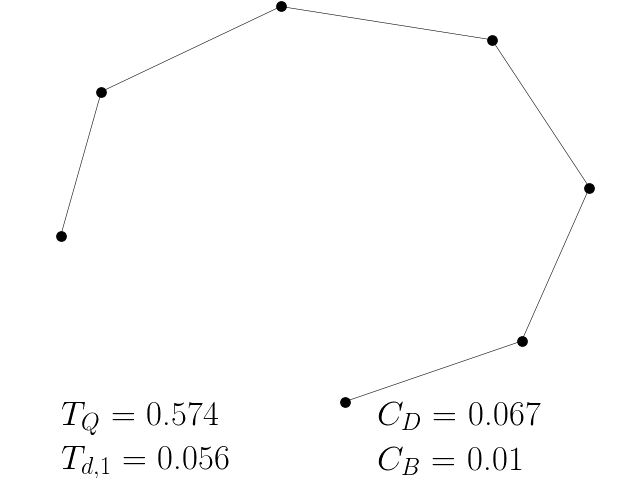} 
        \caption{Path graph.}
    \end{subfigure}
    \hfill   
    \begin{subfigure}{0.19\textwidth}
        \centering
        \includegraphics[width=1\textwidth]{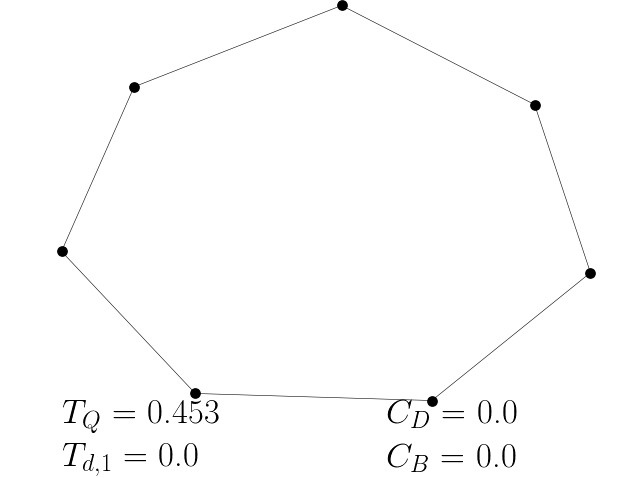} 
        \caption{Circle graph.}
    \end{subfigure} 
    \begin{subfigure}{0.19\textwidth}
        \centering
        \includegraphics[width=1\textwidth]{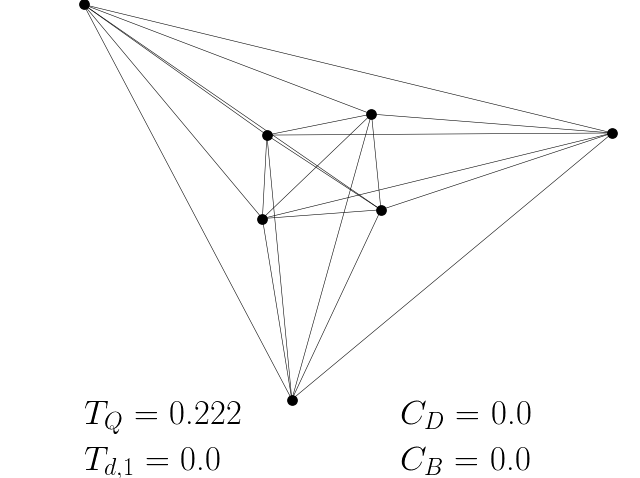} 
        \caption{Complete graph\label{fig:completegraphCD}}
    \end{subfigure}
    \caption{$n=7$ vertex graphs arranged according to decreasing order of the degree centralization measure $C_D$ proposed by Freeman \eqref{eq:Freeman}. \label{fig:CD}}
\end{figure}

\begin{figure}
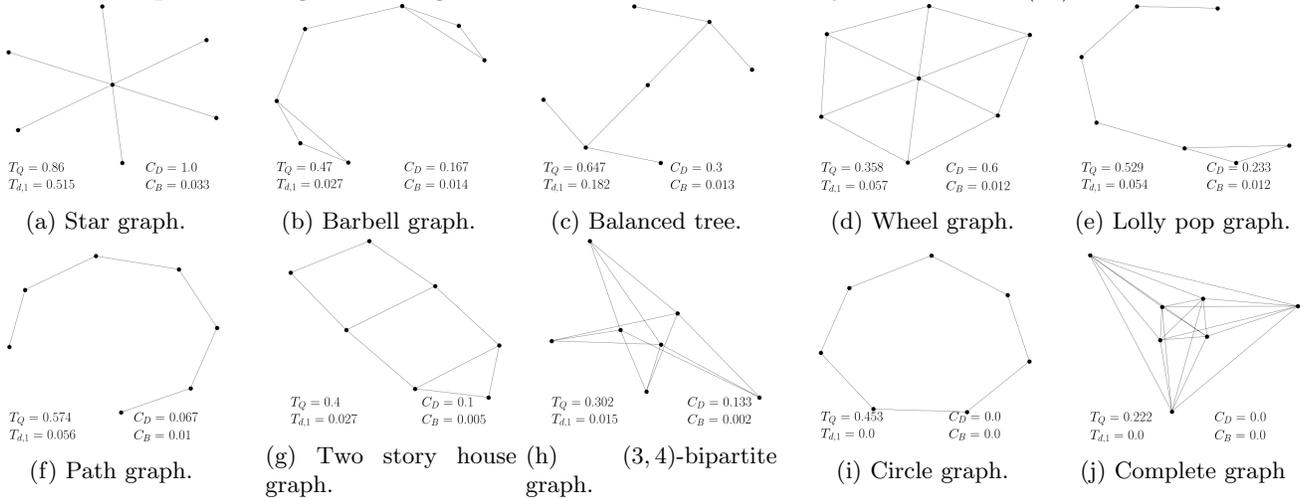

 \textbf{Graph ordering according to the betweenness centrality measure $C_B$  \eqref{eq:FreemanCB}}\\
    \centering
    \begin{subfigure}{0.19\textwidth}
        \centering
        \includegraphics[width=1\textwidth]{figures/star.png} 
        \caption{Star graph. \label{fig:stargraphCB}}
    \end{subfigure}\hfill 
    \begin{subfigure}{0.19\textwidth}
        \centering
        \includegraphics[width=1\textwidth]{figures/barbell.png} 
        \caption{Barbell graph.}
    \end{subfigure}\hfill 
    \begin{subfigure}{0.19\textwidth}
        \centering
        \includegraphics[width=1\textwidth]{figures/balancedTree.png} 
        \caption{Balanced tree.}
    \end{subfigure} 
    \hfill  
     \begin{subfigure}{0.19\textwidth}
        \centering
        \includegraphics[width=1\textwidth]{figures/wheel.png} 
        \caption{Wheel graph.}
    \end{subfigure}
    \begin{subfigure}{0.19\textwidth}
        \centering
        \includegraphics[width=1\textwidth]{figures/lollypop.png} 
        \caption{Lolly pop graph.}
    \end{subfigure}
    \begin{subfigure}{0.19\textwidth}
        \centering
        \includegraphics[width=1\textwidth]{figures/path.png} 
        \caption{Path graph.}
    \end{subfigure}
    \hfill     
     \begin{subfigure}{0.19\textwidth}
        \centering
        \includegraphics[width=1\textwidth]{figures/two_story_house_graph.png} 
        \caption{Two story house graph.}
    \end{subfigure} 
    \begin{subfigure}{0.19\textwidth}
        \centering
        \includegraphics[width=1\textwidth]{figures/comp_multi_bipartite.png} 
        \caption{$(3,4)$-bipartite graph.}
    \end{subfigure}\hfill
    \begin{subfigure}{0.19\textwidth}
        \centering
        \includegraphics[width=1\textwidth]{figures/circle.png} 
        \caption{Circle graph.}
    \end{subfigure} 
    \begin{subfigure}{0.19\textwidth}
        \centering
        \includegraphics[width=1\textwidth]{figures/completeGraph.png} 
        \caption{Complete graph\label{fig:completegraphCB}}
    \end{subfigure}
    \caption{$n=7$ vertex graphs arranged according to decreasing order of the betweenness centralization measure $C_B$  proposed by Freeman \eqref{eq:FreemanCB}. \label{fig:CB}}
\end{figure}

\begin{figure}
    \centering
    \textbf{Graph ordering according to the Theil index $T_{d,1}$ \eqref{eq:TheildegreeEnt}}\\
    \begin{subfigure}{0.19\textwidth}
        \centering
        \includegraphics[width=1\textwidth]{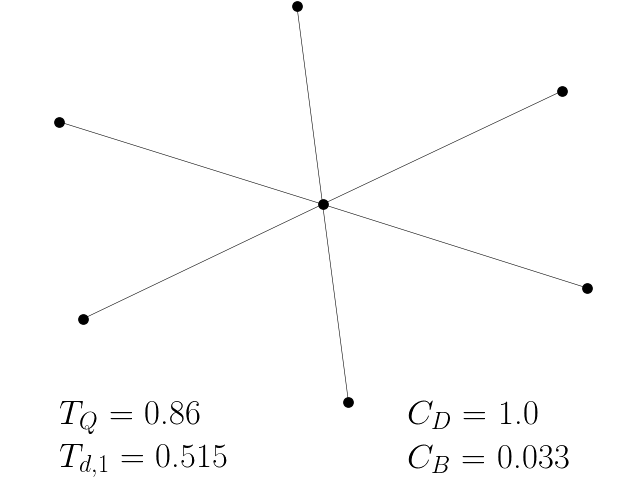} 
        \caption{Star graph. \label{fig:stargraph}}
    \end{subfigure}\hfill 
    \begin{subfigure}{0.19\textwidth}
        \centering
        \includegraphics[width=1\textwidth]{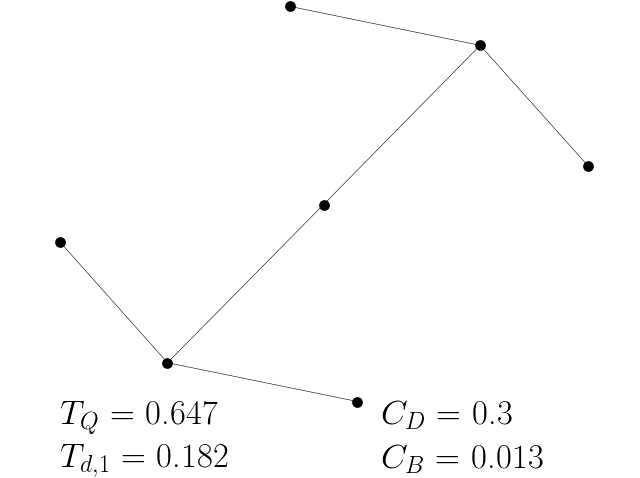} 
        \caption{Balanced tree.}
    \end{subfigure} 
    \hfill  
    \begin{subfigure}{0.19\textwidth}
        \centering
        \includegraphics[width=1\textwidth]{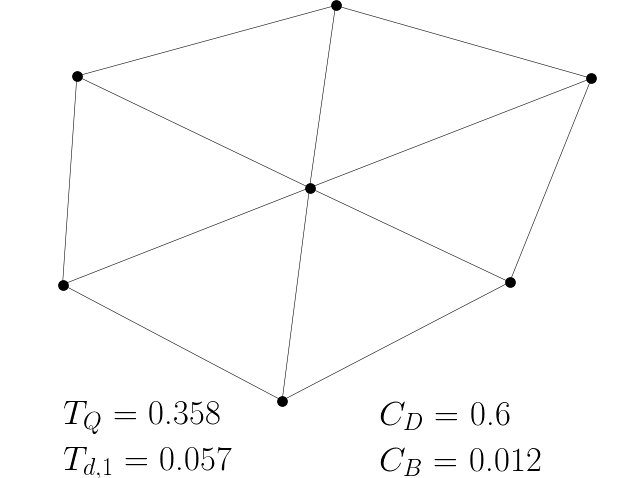} 
        \caption{Wheel graph.}
    \end{subfigure}
    \begin{subfigure}{0.19\textwidth}
        \centering
        \includegraphics[width=1\textwidth]{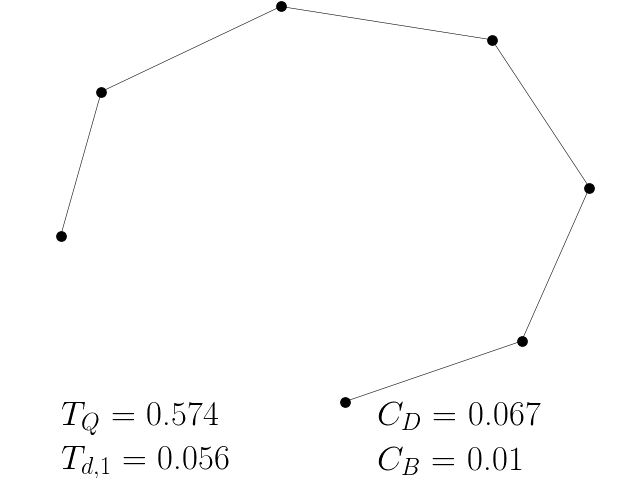} 
        \caption{Path graph.}
    \end{subfigure}
    \hfill   
    \begin{subfigure}{0.19\textwidth}
        \centering
        \includegraphics[width=1\textwidth]{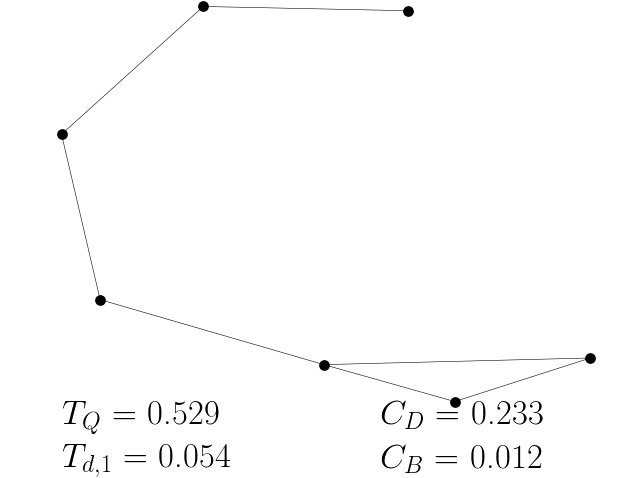} 
        \caption{Lolly pop graph.}
    \end{subfigure}  
    \begin{subfigure}{0.19\textwidth}
        \centering
        \includegraphics[width=1\textwidth]{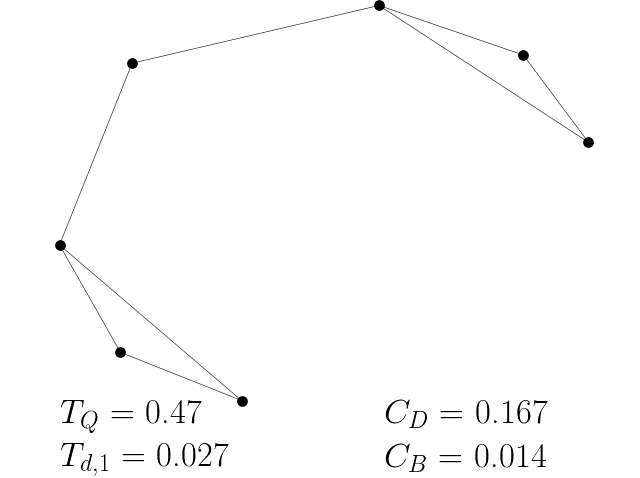} 
        \caption{Barbell graph.}
    \end{subfigure}\hfill 
    \begin{subfigure}{0.19\textwidth}
        \centering
        \includegraphics[width=1\textwidth]{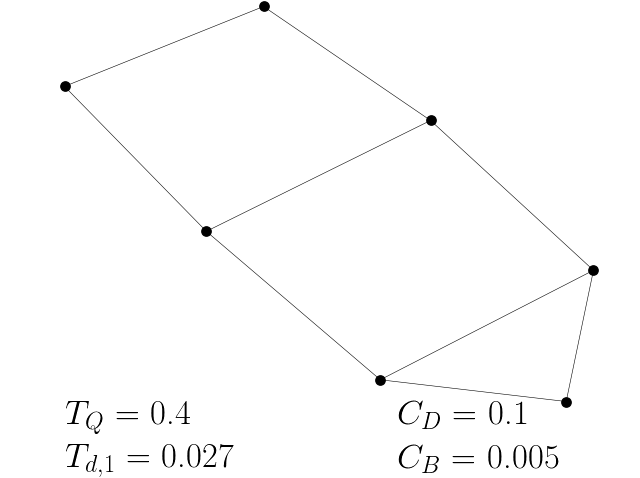} 
        \caption{Two story house graph.}
    \end{subfigure}  
    \begin{subfigure}{0.19\textwidth}
        \centering
        \includegraphics[width=1\textwidth]{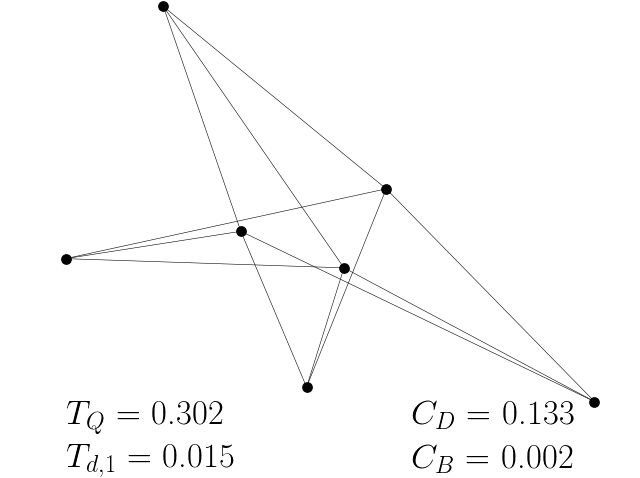} 
        \caption{$(3,4)$-bipartite graph.}
    \end{subfigure}\hfill
    \begin{subfigure}{0.19\textwidth}
        \centering
        \includegraphics[width=1\textwidth]{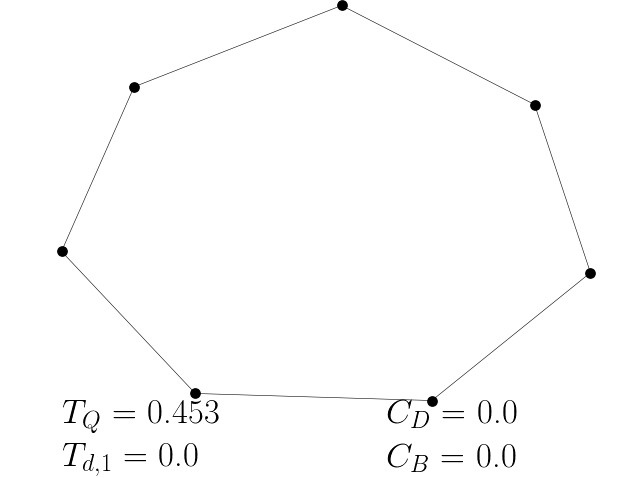} 
        \caption{Circle graph.}
    \end{subfigure}
    \begin{subfigure}{0.19\textwidth}
        \centering
        \includegraphics[width=1\textwidth]{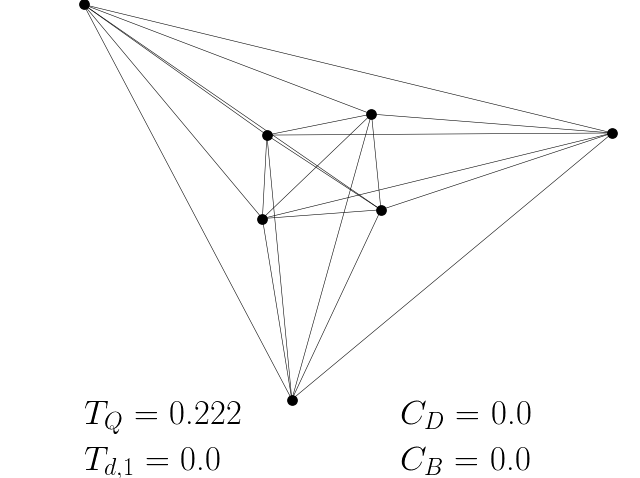} 
        \caption{Complete graph\label{fig:completegraph}}
    \end{subfigure}
    \caption{$n=7$ vertex graphs arranged according to decreasing order of the Theil index $T_{d,1}$ \eqref{eq:TheildegreeEnt}.\label{fig:TI}}
\end{figure}

 \begin{figure}
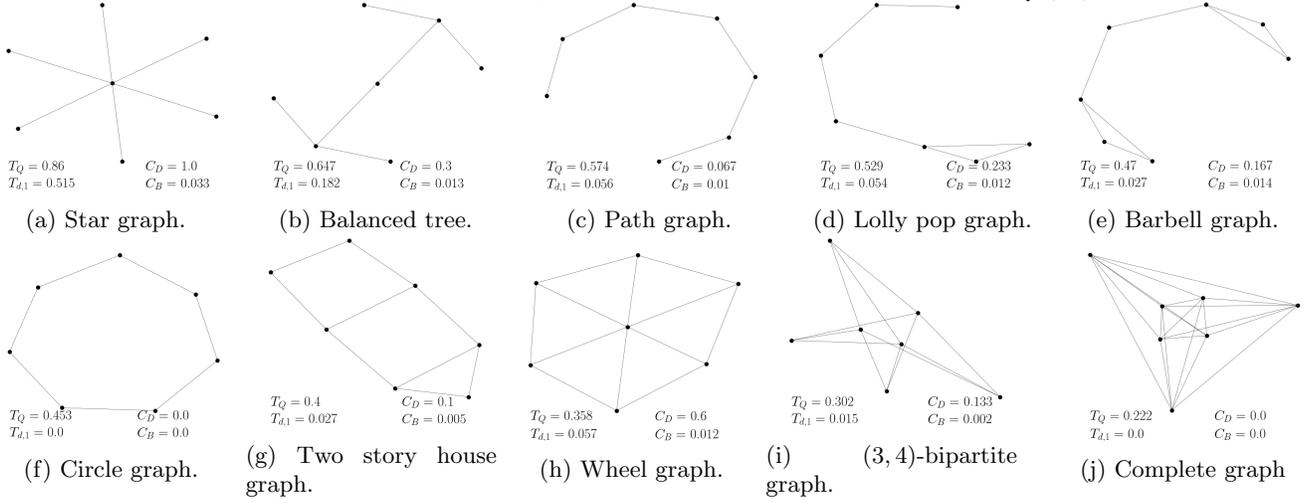

    \centering
    \textbf{Graph ordering according to the von Neumann Theil index $T_Q$ \eqref{eq:Theilquant}}\\
    \begin{subfigure}{0.19\textwidth}
        \centering
        \includegraphics[width=1\textwidth]{figures/star.png} 
        \caption{Star graph. \label{fig:stargraph1}}
    \end{subfigure}\hfill 
    \begin{subfigure}{0.19\textwidth}
        \centering
        \includegraphics[width=1\textwidth]{figures/balancedTree.png} 
        \caption{Balanced tree.}
    \end{subfigure} 
    \hfill  
    \begin{subfigure}{0.19\textwidth}
        \centering
        \includegraphics[width=1\textwidth]{figures/path.png} 
        \caption{Path graph.}
    \end{subfigure}
    \hfill   
    \begin{subfigure}{0.19\textwidth}
        \centering
        \includegraphics[width=1\textwidth]{figures/lollypop.png} 
        \caption{Lolly pop graph.}
    \end{subfigure}
    \begin{subfigure}{0.19\textwidth}
        \centering
        \includegraphics[width=1\textwidth]{figures/barbell.png} 
        \caption{Barbell graph.}
    \end{subfigure}\hfill
    \begin{subfigure}{0.19\textwidth}
        \centering
        \includegraphics[width=1\textwidth]{figures/circle.png} 
        \caption{Circle graph.}
    \end{subfigure} 
    \begin{subfigure}{0.19\textwidth}
        \centering
        \includegraphics[width=1\textwidth]{figures/two_story_house_graph.png} 
        \caption{Two story house graph.}
    \end{subfigure}  
    \begin{subfigure}{0.19\textwidth}
        \centering
        \includegraphics[width=1\textwidth]{figures/wheel.png} 
        \caption{Wheel graph.}
    \end{subfigure}   
    \begin{subfigure}{0.19\textwidth}
        \centering
        \includegraphics[width=1\textwidth]{figures/comp_multi_bipartite.png} 
        \caption{$(3,4)$-bipartite graph.}
    \end{subfigure}\hfill
    \begin{subfigure}{0.19\textwidth}
        \centering
        \includegraphics[width=1\textwidth]{figures/completeGraph.png} 
        \caption{Complete graph\label{fig:completegraph1}}
    \end{subfigure}
    \caption{$n=7$ vertex graphs arranged according to decreasing order of the von Neumann Theil index $T_Q$ \eqref{eq:Theilquant}.\label{fig:QTI}}
\end{figure}

\section{Future Direction}

In the previous sections, we demonstrated that the von Neumann index of
a graph could be used to bound the graph's Theil index (Theorem \ref{lem:degreevon}). The Theil
index was shown to be a canonical (but crude) measure of centralization (sections \ref{sec:graphTheil} and \ref{sec:TheilandCent}).
We therefore concluded that the graph's von Neumann index (specifically,
its \emph{von Neumann Theil index}) provided key structural information
about the level of centralization present across the graph. 
In this section, we show that the von Neumann Theil index can be generalized
by considering the R\'enyi entropy. We call this the generalized
von Neumann Theil index. The generalized von Neumann Theil index is shown
to bound the negative logarithm of the graph's Jain fairness index.

\subsection{Generalizing the von Neumann Theil index}

Let $\rho_{G}$ be the graph's trace normalized Laplacian \eqref{eq:vertexstate}. Above (Definition \ref{def:quantumTheil}),
the von Neumann Theil index was defined to be 
\begin{equation}
T_{Q}\left(G\right)=\log n-H\left(G\right),\label{eq:quantumTheilIndex}
\end{equation}
where 
\[
H\left(G\right)=-\mathrm{Tr}\left\{ \rho_{G}\log\rho_{G}\right\} .
\]
Equation (\ref{eq:quantumTheilIndex}) can be trivially generalized
by considering the Renyi entropy \cite{muller2013quantum} of $\rho_{G}$, which is defined to be
\[
H^{\left(p\right)}:=\frac{1}{1-p}\log\mathrm{Tr}\left\{ \rho_{G}^{p}\right\} .
\]
Using L'Hopital's rule \cite{waner1995introduction}, we find that
\begin{equation}
H\left(G\right)=\lim_{p\to1}H^{\left(p\right)}.\label{eq:renyivoneq}
\end{equation}
We can therefore define the graph's generalized von Neumann Theil index
to be 
\begin{equation}
T_{Q}^{\left(p\right)}:=\log n-H^{\left(p\right)}\label{eq:genQTheilInd}
\end{equation}
and from \eqref{eq:renyivoneq} we have 
\begin{equation}
T_{Q}\left(G\right)=\lim_{p\to1}T_{Q}^{\left(p\right)}\left(G\right).\label{eq:renyivoneq-1}
\end{equation}
Note, because $H^{(p)}$ is a monotonically decreasing function of $p$, the generalized von Neumann Theil index is a monotonically increasing function of $p$.

\subsection{Bounding the Jain Fairness Index}

In the following lemma, we show that \eqref{eq:genQTheilInd} can be used to bound the negative logarithm of the Jain fairness index \cite{jain1984quantitative} (an index used to measure the \emph{fairness} of a distribution of characteristics among a set of nodes). This adds further credit to the idea that von Neumann entropy (and its generalization) characterises graph centralization.

\begin{lem}
Let $G$ be a connected graph on $n$ vertices, where the degree of
vertex $i$ is given by $d_{i}$. Define the Jain fainess index to
be \cite{jain1984quantitative}
\begin{equation}
\mathcal{J}\left(d_{1},\dots,d_{n}\right)=\frac{\left(\sum d_{i}\right)^{2}}{n\sum d_{i}^{2}}.\label{eq:jainindex}
\end{equation}
Then the following equality holds:
\[
T_{Q}^{\left(2\right)}\geq-\log\mathcal{J}\left(d_{1},\dots,d_{n}\right).
\]
\end{lem}
\begin{proof}
This follows from
\begin{align*}
T_{Q}^{\left(2\right)}=\log n-\log\left(\frac{\left(\sum d_{i}\right)^{2}}{\sum d_{i}+\sum d_{i}^{2}}\right) & \geq\log n-\log\left(\frac{\left(\sum d_{i}\right)^{2}}{\sum d_{i}^{2}}\right)\\
 & =-\log\left(\frac{\left(\sum d_{i}\right)^{2}}{n\sum d_{i}^{2}}\right)\\
 & =-\log\mathcal{J}\left(d_{1},\dots,d_{n}\right).
\end{align*}
The first equality follows from \cite[equation (2)]{dairyko2016note}
and the last line follows from (\ref{eq:jainindex}).
\end{proof}

\section{Conclusions}

 In this work, we define an $n$ point distribution on an $n$ vertex graph (which we call the \emph{relative} degree distribution), where each point of the distribution is proportional to a positive power of the corresponding vertex degree. We then use this distribution to construct a Theil index for the graph. Naturally, this Theil index can be interpreted as a measure of graph centralization. We then show that the von Neumann index of the graph's trace normalized combinatorial Laplacian can be used to bound the graph's Theil index. This is the first time that von Neumann index has been related to the Theil index. We refer to the bound as the von Neumann Theil index. 
 We conclude that the graph's von Neumann index (specifically, the von Neumann Theil index) provides  structural information about the level of centralization within the graph. We also demonstrate that the von Neumann Theil index provides a more precise measure of centralization when compared to traditional centralization measures and the graph's classical Theil index. As a final note, we demonstrate that the von Neumann Theil index can be generalized by considering the quantum R\'enyi-p entropy. We show that with $p=2$, this generalization bounds the negative logarithm of the graph's Jain fairness index, further demonstrating the relationship between centralization and quantum graph entropy.
 
 \section*{Funding}
This work was supported by the EPSRC under grant number (EP/K04057X/2) and the UK National Quantum Technologies Programme under grant number (EP/M013243/1).

\bibliographystyle{unsrt}
\bibliography{animeshrefs}

\end{document}